\documentclass{article}
\usepackage{times}

\usepackage{amsmath,amssymb,bbold,bm,amsfonts,nicefrac}
\usepackage{amsthm,framed,url}
\usepackage{graphicx}
\usepackage{tikz}
\usepackage{color}
\usepackage{psfrag}
\usepackage{epstopdf}
\usepackage{wrapfig}
\usepackage{enumitem,array}
\usepackage{cite}
\usepackage{hyperref}
\usepackage{algorithm}
\usepackage{algorithmic}

\usepackage{textcomp,balance}

\newtheorem{theorem}{Theorem}

\newtheorem{lemma}{Lemma}

\theoremstyle{remark}
\newtheorem*{Remark}{Remark}

\newcommand{\y}{\boldsymbol{y}}
\newcommand{\z}{\boldsymbol{z}}
\newcommand{\p}{\boldsymbol{p}}
\newcommand{\q}{\boldsymbol{q}}

\newcommand{\blambda}{\boldsymbol{\lambda}}

\newcommand{\PI}{\boldsymbol{\varPi}}
\newcommand{\THETA}{\boldsymbol{\varTheta}}
\newcommand{\ones}{\boldsymbol{\mathit{1}}}

\newcommand{\Y}{\boldsymbol{Y}}
\renewcommand{\P}{\boldsymbol{P}}

\newcommand{\bR}{\mathbb{R}}

\newcommand{\prob}{\mathbb{P}}

\newcommand{\T}{{\!\top\!}}
\newcommand{\mode}[1]{{(\!#1\!)\!}}
\newcommand{\cpd}[1]{\text{\textlbrackdbl}#1\text{\textrbrackdbl}}

\DeclareMathOperator*{\minimize}{\textup{minimize}}

\newcommand{\st}{\text{subject to}}

\title{Kullback-Leibler Principal Component for Tensors is not NP-hard}

\author{
{Kejun Huang}%
\thanks{University of Minnesota, Minneapolis, MN 55455. Email: \texttt{huang663@umn.edu}}
\and
{Nicholas D. Sidiropoulos}%
\thanks{University of Virginia, Charlottesville, VA 22904. Email: \texttt{nikos@virginia.edu}}
}

\date{}

\begin{document}
\maketitle

\begin{abstract}
We study the problem of nonnegative rank-one approximation of a nonnegative tensor, and show that the globally optimal solution that minimizes the generalized Kullback-Leibler divergence can be efficiently obtained, i.e., it is not NP-hard. This result works for arbitrary nonnegative tensors with an arbitrary number of modes (including two, i.e., matrices). We derive a closed-form expression for the KL principal component, which is easy to compute and has an intuitive probabilistic interpretation. For generalized KL approximation with higher ranks, the problem is for the first time shown to be equivalent to multinomial latent variable modeling, and an iterative algorithm is derived that resembles the expectation-maximization algorithm. On the Iris dataset, we showcase how the derived results help us learn the model in an \emph{unsupervised} manner, and obtain strikingly close performance to that from supervised methods.
\end{abstract}

\section{Introduction}

Tensors are powerful tools for big data analytics~\cite{sidiropoulos2017tensor}, mainly thanks to the ability to \emph{uniquely} identify latent factors under mild conditions~\cite{kruskal1977three,sidiropoulos2000uniqueness}. On the other hand, most detection and estimation problems related to tensors are NP-hard~\cite{hillar2013most}. A similar situation is encountered in nonnegative matrix factorization, which is essentially unique under certain conditions~\cite{huang2014non,fu2017identifiability} and computationally NP-hard~\cite{Vavasis2009}. In a lot of applications, nonnegativity constraints are natural for tensor latent factors as well.

In practice, the latent factors of tensors and matrices are usually obtained by minimizing the mismatch between the data and the factorization model according to certain loss measures. The most popular loss measures the sum of the element-wise squared errors, which is conceptually appealing and conducive for algorithm design, thanks to the success of least-squares-based methods. For example, an effective algorithm for minimizing the least-squares loss is AO-ADMM~\cite{huang2016flexible}, and we refer the readers to the references therein for other least-squares-based methods. From an estimation theoretic point of view, the least-squares loss admits a maximum-likelihood interpretation under i.i.d. Gaussian noise. In a lot of applications, however, it remains questionable whether Gaussian noise is a suitable model for \emph{nonnegative} data.

We study the problem of fitting a nonnegative data matrix/tensor $\Y$ with low rank factors, using the generalized Kullback-Leibler (KL) divergence as the fitting criterion. Mathematically, given a $N$-way tensor data $\Y\in\bR^{J_1\times...\times J_N}$ and a target rank $K$, we try to find factor matrices constituting a canonical polyadic decomposition (CPD) that best approximates the data tensor $\Y$ in terms of generalized KL divergence:

\begin{equation}\label{prob:kl-ncp}
\begin{aligned}
\minimize_{\blambda,\{\P^\mode{n}\}} & 
	\sum_{j_1,...,j_N} \left(-Y_{j_1...j_N}\log\sum_{k=1}^{K}\lambda_k\prod_{n=1}^{N}P^\mode{n}_{j_nk} \right.\\
&	\left.\hspace*{80pt}	+ \sum_{k=1}^{K}\lambda_k\prod_{n=1}^{N}P^\mode{n}_{j_nk} \right) \\
\st & ~~~\blambda\geq0, 
\P^\mode{n}\geq0, \ones^\T\P^\mode{n}=\ones^\T, \forall~n\in[N].
\end{aligned}
\end{equation}
The conditions imposed in Problem~\eqref{prob:kl-ncp} besides nonnegativity are intended for resolving the trivial scaling ambiguity inherent in matrix factorization and tensor CPD models, and thus are without loss of generality: the columns of all the factor matrices are normalized to sum up to one, and the scalings are absorbed into the corresponding values in the vector $\blambda$.  We adopt the common notation $\cpd{\blambda;\P^\mode{1},...,\P^\mode{N}}$ to denote the tensor synthesized from the CPD model using these factors.
%and denote the loss function in~\eqref{prob:kl-ncp} as $\KL(\Y\|\cpd{\blambda;\P^\mode{1},...,\P^\mode{N}})$. 

An instant advantage of this formulation is that the loss function of~\eqref{prob:kl-ncp} can be greatly simplified. Thanks to the constraints 
$\ones^\T\P^\mode{n}=\ones^\T$,
it is easy to see that
\[
\sum_{j_1,...,j_N}\sum_{k=1}^{K}\lambda_k\prod_{n=1}^{N}P^\mode{n}_{j_nk} = \sum_{k=1}^{K}\lambda_k.
\]
Therefore, Problem~\eqref{prob:kl-ncp} is mathematically equivalent to
\begin{equation}\label{prob:kl-ncp0}
\begin{aligned}
\minimize_{\blambda,\{\P^\mode{n}\}} & ~~
	-\hspace*{-5pt}\sum_{j_1,...,j_N} Y_{j_1...j_N}\log\sum_{k=1}^{K}\lambda_k\prod_{n=1}^{N}P^\mode{n}_{j_nk} 
	+ \ones^\T\blambda\\
\st & ~~~\blambda\geq0, 
\P^\mode{n}\geq0, \ones^\T\P^\mode{n}=\ones^\T, \forall~n\in[N].
\end{aligned}
\end{equation}
%
%\begin{figure}
%\centering
%\input{figs/cpd}
%\caption{The canonical polyadic decomposition (CPD) of a 3-way tensor.}
%\label{fig:cpd}
%\end{figure}

\section{Motivation: Probabilistic Latent Variable Modeling}

Most of the existing literature motivates the use of generalized KL-divergence by modeling the nonnegative integer data as generated from Poisson distributions~\cite{Chi2012}. Specifically, the model states that each entry of the tensor $Y_{j_1...j_N}$ is generated from a Poisson distribution with parameter $\varTheta_{j_1...j_N}$, and the underlying tensor $\THETA$ admits an exact CPD model
\(
\THETA = \cpd{\blambda;\P^\mode{1},...,\P^\mode{N}}.
\)
While it is a simple and reasonable model, the physical meaning behind the CPD model for the underlying Poisson parameters is not entirely clear.

In this paper we give the choice of generalized KL-divergence as the loss function a more compelling motivation. Consider $N$ categorical random variables $X_1,...,X_N$, each taking $J_1,...,J_N$ possible outcomes, respectively. Suppose these random variables are jointly drawn for a number of times, each time independently, and the outcome counts are recorded in a $N$-dimensional count tensor $\Y$, which is the data we are given. Denote the joint probability that $X_1=j_1,...,X_N=j_N$ as $\varPi_{j_1...j_N}$, i.e.
\[
\prob\left[ X_1=j_1,...,X_N=j_N \right] = \varPi_{j_1...j_N},
\]
where we have
\[
\sum_{j_1,...,j_N} \varPi_{j_1...j_N} = 1.
\]
Then the overall probability that, out of $M$ independent draws, the event of $X_1=j_1,...,X_N=j_N$ occurs $Y_{j_1...j_N}$ times is
\[
\prob\left[\Y\right] = \frac{M!}{\prod\left(Y_{j_1...j_N}!\right)}
	\prod_{j_1,...,j_N} \left(\varPi_{j_1...j_N}\right)^{Y_{j_1...j_N}},
\]
where
\begin{equation}\label{eq:M}
\sum_{j_1,...,j_N} Y_{j_1...j_N} = M.
\end{equation}
The maximum likelihood estimates of the parameters are simply
\[
\varPi_{j_1...j_N} = \frac{Y_{j_1...j_N}}{M}.
\]
However, this simple estimate of $\PI$ may not be of much practical use. First of all, the number of parameters we are trying to estimate is the same as the number of possible outcomes, which is not a parsimonious model; in other words, we are not exploiting any possible structure between the random variables $X_1,...,X_N$, other than the fact that exists a joint distribution between them.
Furthermore, as a result of over-parameterization, we will need the number of independent draws $M$ to be very large before we can have an accurate estimate of $\PI$, which is often not the case in practice.

A simple and widely used assumption we can impose onto the set of variables is the naive Bayes model: Suppose there is a hidden random variable $\varXi$, which is also categorical and can take $K$ possible outcomes, such that $X_1,...,X_N$ are conditionally independent given $\varXi$. The corresponding graphical model is given in Fig.~\ref{fig:gm}. Mathematically, this means
\[
\prob\left[
X_1=j_1,...,X_N=j_N | \varXi=k
\right] = 
\prod_{n=1}^N \prob\left[ X_n=j_n | \varXi=k \right].
\]
Then we have
\begin{align*}
& \prob\left[ X_1=j_1,...,X_N=j_N \right] \\
& = \sum_{k=1}^{K}\prob\left[ X_1=j_1,...,X_N=j_N | \varXi=k \right] \prob[\varXi=k] \\
& = \sum_{k=1}^{K} \prob[\varXi=k] \prod_{n=1}^N \prob\left[ X_n=j_n | \varXi=k \right].
\end{align*}
Denote
\[
\overline{\lambda}_k = \prob[\varXi=k] 
\text{~~and~~} 
P^\mode{n}_{j_nk} = \prob\left[ X_n=j_n | \varXi=k \right],
\]
then it is easy to see that
\begin{equation}\label{eq:Pi=cpd}
\PI = \cpd{\overline{\blambda};\P^\mode{1},...,\P^\mode{N}},
\end{equation}
which means $\PI$ admits an exact CPD model.
Even though naive Bayes seems to be a very specific model, it has recently been shown that it is far more general than meets the eye---no matter how dependent $X_1,...,X_N$ are, there always exists a hidden variable $\varXi$ such that the depicted naive Bayes model holds,  thanks to a link between tensors and probability established in \cite{kargas2017completing}.

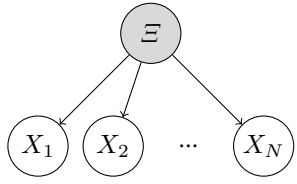
\begin{figure}[t]
\centering
\begin{tikzpicture}
\node at  (-0.5,1.5) [circle,draw,fill=gray!30,minimum size=8mm] (Xi) {$\varXi$};
\node at (-2,0) [circle,draw,minimum size=8mm,inner sep=-5pt] (X1) {$X_1$};
\node at (-1,0) [circle,draw,minimum size=8mm,inner sep=-5pt] (X2) {$X_2$};
\node at ( 0,0) {...};
\node at ( 1,0) [circle,draw,minimum size=8mm,inner sep=-5pt] (XN) {$X_N$};
\draw [->] (Xi) -- (X1);
\draw [->] (Xi) -- (X2);
\draw [->] (Xi) -- (XN);
\end{tikzpicture}
\caption{A graphical model depicting the probabilistic dependencies between the hidden variable $\varXi$ and the observed variables $X_1,...,X_N$, in a naive Bayes model.}
\label{fig:gm}
\end{figure}

Using the CPD parameterization of the multinomial parameter $\PI$, we formulate the maximum likelihood estimation of $\prob[\varXi=k]$ and 
$\prob[X_n=j_n|\varXi=k]$ as the following optimization problem:
\begin{equation}\label{prob:kl-ncp1}
\begin{aligned}
\minimize_{\overline{\blambda},\{\P^\mode{n}\}} & 
	~~-\hspace*{-5pt}\sum_{j_1,...,j_N} Y_{j_1...j_N}\log\sum_{k=1}^{K}\overline{\lambda}_k\prod_{n=1}^{N}P^\mode{n}_{j_nk}  \\
\st & ~~~\overline{\blambda}\geq0, \ones^\T\overline{\blambda}=1, \\
& ~~~ \P^\mode{n}\geq0, \ones^\T\P^\mode{n}=\ones^\T, \forall~n\in[N].
\end{aligned}
\end{equation}
Problem~\eqref{prob:kl-ncp1} is different from Problem~\eqref{prob:kl-ncp0}, but the difference is small---in~\eqref{prob:kl-ncp1}, $\overline{\blambda}$ is constrained to sum up to one, whereas in~\eqref{prob:kl-ncp0}, the sum of the elements of $\blambda$ is penalized in the loss function. In fact, the two problems are exactly equivalent, despite their apparent differences.

\begin{theorem}\label{thm:equivalence}
Let $(\blambda_\star,\{\P^\mode{n}_\star\})$ be an optimal solution for Problem~\eqref{prob:kl-ncp0}, then $(\overline{\blambda}_\star,\{\P^\mode{n}_\star\})$ is an optimal solution for Problem~\eqref{prob:kl-ncp1}, where
\[
\overline{\blambda}_\star = \frac{1}{\ones^\T\blambda_\star}\blambda_\star.
\]
\end{theorem}

Before proving Theorem~\ref{thm:equivalence}, we first show the following lemma, which is interesting in its own right.

\begin{lemma}\label{lmm:sum}
If $\blambda_\star$ is optimal for~\eqref{prob:kl-ncp0}, then
\[
\ones^\T\blambda = \sum_{j_1,...,j_N} Y_{j_1...j_N}.
\]
\end{lemma}
\begin{proof}
We show this by checking the optimality condition of~\eqref{prob:kl-ncp0} with respect to $\blambda$. Without loss of generality, we can assume that $\blambda_\star>0$ strictly, because otherwise the rank can be reduced. Since the inequality constraints with respect to $\blambda$ are not attained as equalities, their corresponding dual variables are equal to zero, according to complementary slackness. The KKT condition for~\eqref{prob:kl-ncp0} with respect to $\blambda$ then reduces to the gradient of the loss function of~\eqref{prob:kl-ncp0} with respect to $\blambda$ at $\blambda_\star$ being equal to zero. Specifically, setting the derivative with respect to $\lambda_k$ equal to zero yields
\[
\sum_{j_1,...,j_N}\frac{Y_{j_1...j_N}}{\sum_\kappa\lambda_{\star\kappa}\prod_\nu P^\mode{\nu}_{j_\nu k}}
\prod_{n=1}^{N}P^\mode{n}_{j_nk} = 1.
\]
Therefore
\begin{align*}
\sum_{k=1}^{K}\lambda_{\star k} 
& = \sum_{k=1}^{K}\left(\lambda_{\star k}
\sum_{j_1,...,j_N}\frac{Y_{j_1...j_N}}{\sum_\kappa\lambda_{\star\kappa}\prod_\nu P^\mode{\nu}_{j_\nu k}}
\prod_{n=1}^{N}P^\mode{n}_{j_nk}\right) \\
& = \sum_{j_1,...,j_N}\frac{Y_{j_1...j_N}}{\sum_\kappa\lambda_{\star\kappa}\prod_\nu P^\mode{\nu}_{j_\nu k}}
\left(\sum_{k=1}^{K}\lambda_{\star k}\prod_{n=1}^{N}P^\mode{n}_{j_nk}\right) \\
& = \sum_{j_1,...,j_N} Y_{j_1...j_N}. \qedhere
\end{align*}
\end{proof}

We now prove Theorem~\ref{thm:equivalence} with the help of Lemma~\ref{lmm:sum}.
\begin{proof}[Proof of Theorem~\ref{thm:equivalence}]
We show that $(\overline{\blambda}_\star,\{\P^\mode{n}_\star\})$ is optimal for~\eqref{prob:kl-ncp1} via contradiction.

Suppose $(\overline{\blambda}_\star,\{\P^\mode{n}_\star\})$ is not optimal for~\eqref{prob:kl-ncp1}, then there exists a feasible point $(\overline{\blambda}_\flat,\{\P^\mode{n}_\flat\})$ such that
\begin{align*}
-\hspace*{-5pt}\sum_{j_1,...,j_N} &
Y_{j_1...j_N}\log\sum_{k=1}^{K}\overline{\lambda}_{\flat k}
\prod_{n=1}^{N}P^\mode{n}_{\flat j_nk}\\
& < -\hspace*{-5pt}\sum_{j_1,...,j_N} 
Y_{j_1...j_N}\log\sum_{k=1}^{K}\overline{\lambda}_{\star k}
\prod_{n=1}^{N}P^\mode{n}_{\star j_nk} 
\end{align*}
Define $\blambda_\flat = M\overline{\blambda}_\flat$, then  $(\blambda_\flat,\{\P^\mode{n}_\flat\})$ is clearly feasible for~\eqref{prob:kl-ncp0}. Furthermore, we have
\begin{align*}
& -\hspace*{-5pt}\sum_{j_1,...,j_N}
Y_{j_1...j_N}\log\sum_{k=1}^{K}\lambda_{\flat k}
\prod_{n=1}^{N}P^\mode{n}_{\flat j_nk} + \ones^\T\blambda_\flat\\
& = -\hspace*{-5pt}\sum_{j_1,...,j_N}
Y_{j_1...j_N}\log\sum_{k=1}^{K}\overline{\lambda}_{\flat k}
\prod_{n=1}^{N}P^\mode{n}_{\flat j_nk} - M\log M + M\\
& < -\hspace*{-5pt}\sum_{j_1,...,j_N} 
Y_{j_1...j_N}\log\sum_{k=1}^{K}\overline{\lambda}_{\star k}
\prod_{n=1}^{N}P^\mode{n}_{\star j_nk} - M\log M + M \\
& = -\hspace*{-5pt}\sum_{j_1,...,j_N}
Y_{j_1...j_N}\log\sum_{k=1}^{K}\lambda_{\flat k}
\prod_{n=1}^{N}P^\mode{n}_{\star j_nk} + \ones^\T\blambda_\star,
\end{align*}
where the equalities stems from Lemma~\ref{lmm:sum}.
This means $(\blambda_\flat,\{\P^\mode{n}_\flat\})$ gives a smaller loss value for~\eqref{prob:kl-ncp0} than that of $(\blambda_\star,\{\P^\mode{n}_\star\})$, and contradicts our assumption that $(\blambda_\star,\{\P^\mode{n}_\star\})$ is optimal for~\eqref{prob:kl-ncp0}.
\end{proof}
A similar but less general result for the case when $N=2$ is given in~\cite{Ho2008}, in the context of nonnegative matrix factorization using generalized KL-divergence loss.

The take home point of this section is that we can find the maximum likelihood estimate of the hidden variable in the naive Bayes model by taking the nonnegative CPD of the data tensor. In a way, our analysis suggests that the generalized KL-divergence is a more suitable loss function to fit nonnegative data than, for example, the $L_2$ loss.

\section{KL Principal Component}

Now that we have established how important Problem~\eqref{prob:kl-ncp0} is in probabilistic latent variable modeling, we focus on a specific case of~\eqref{prob:kl-ncp0} when $K=1$. This case corresponds to extracting the ``principal component'' of a nonnegative tensor under the generalized KL-divergence loss. We will show, in this section, that this specific problem is not only computationally tractable, but admits an extremely simple closed form solution.

Let us first rewrite Problem~\eqref{prob:kl-ncp0} with $K=1$. In this case, the diagonal loadings $\blambda$ become a scalar $\lambda$, and the individual factor matrices $\P^\mode{n}$ become vectors $\p^\mode{n}$:
\begin{equation}\label{prob:kl-pc}
\begin{aligned}
\minimize_{\blambda,\{\P^\mode{n}\}} & ~~
	-\hspace*{-5pt}\sum_{j_1,...,j_N} Y_{j_1...j_N}\log\lambda\prod_{n=1}^{N}p^\mode{n}_{j_n} 
	+ \lambda\\
\st & ~~~\lambda\geq0, 
\p^\mode{n}\geq0, \ones^\T\p^\mode{n}=1, \forall~n\in[N].
\end{aligned}
\end{equation}
A salient feature of this case when $K=1$ is that there is no summation in the $\log$, just a product. Therefore, we can equivalently write Problem~\eqref{prob:kl-pc} as
\begin{equation}\label{prob:kl-pc_convex}
\begin{aligned}
\minimize_{\blambda,\{\P^\mode{n}\}} & ~~
	-\hspace*{-5pt}\sum_{j_1,...,j_N} Y_{j_1...j_N}
	\left(\log\lambda + \sum_{n=1}^{N}\log p^\mode{n}_{j_n} \right)
	+ \lambda\\
\st & ~~~\lambda\geq0, 
\p^\mode{n}\geq0, \ones^\T\p^\mode{n}=1, \forall~n\in[N].
\end{aligned}
\end{equation}
Noticing that $-\log$ is a convex function, the exciting observation we see in Problem~\eqref{prob:kl-pc_convex} is that it is in fact convex! This already means that it can be solved optimally and efficiently~\cite{boyd2004convex}, but we will in fact show a lot more: that it admits very simple and intuitive {\em closed-form} solution.

\begin{Remark}
It may seem obvious from our derivation that once we try rewriting Problem~\eqref{prob:kl-ncp0} with $K=1$, one will immediately see that this problem has hidden convexity in it. However, to the best of our knowledge, we are the first to point out this fact. In hindsight, there is a subtle caveat that plays a key role in spotting this hidden convexity: we fixed the inherent scaling ambiguities by constraining all the columns of the factor matrices to sum up to one; as a result, the sum of all the entries of the reconstructed tensor, which appears in the generalized KL-divergence loss, boils down to simply the sum of the diagonal loadings. If this trick is not applied to the problem formulation, there is still a multi-linear term in the loss of~\eqref{prob:kl-pc_convex}, and the hidden convexity is not at all obvious.
\end{Remark}

We now derive an optimal solution for Problem~\eqref{prob:kl-pc_convex} by checking the KKT conditions. For convex problems such as~\eqref{prob:kl-pc_convex}, KKT condition is not only necessary, but also sufficient for optimality.
Using Lemma~\ref{lmm:sum}, we immediately have that
\begin{equation*}
\lambda = \sum_{j_1,...,j_N} Y_{j_1...j_N} = M.
\end{equation*}
For the variable $\p^\mode{n}$, let $\xi^\mode{n}$ denote the dual variable corresponding to the equality constraint $\ones^\T\p^\mode{n}=1$, and $\q^\mode{n}$ denote the nonnagative dual variable corresponding to the inequality constraint $\p^\mode{n}\geq0$, for all $n=1,...,N$. 
The loss function in~\eqref{prob:kl-pc_convex} separates down to the components, and the term corresponding to $p^\mode{n}_{j_n}$ is
\[
-\sum_{\substack{j_1,...,j_{n-1},\\j_{n+1},...,j_N}} Y_{j_1...j_N}
\log p^\mode{n}_{j_n}
= -y^\mode{n}_{j_n}\log p^\mode{n}_{j_n},
\]
where we denote
\begin{equation}\label{eq:y}
y^\mode{n}_{j_n} = \sum_{\substack{j_1,...,j_{n-1},\\j_{n+1},...,j_N}} Y_{j_1...j_N}.
\end{equation}
We see that $p^\mode{n}_{j_n}$ cannot be equal to zero if $y^\mode{n}_{j_n}\neq0$, because otherwise it will drive the loss value up to $+\infty$; according to complementary slackness, this means the corresponding $q^\mode{n}_{j_n}=0$. If $y^\mode{n}_{j_n}=0$, then $p^\mode{n}_{j_n}$ does not directly affect the loss value of~\eqref{prob:kl-pc_convex}, even when it equals to zero, using the convention that $0\log0=0$. Since we have the constraint $\ones^\T\p^\mode{n}=1$, such a $p^\mode{n}_{j_n}$ should be equal to zero at optimality, otherwise the other entries in $\p^\mode{n}$ will be smaller, leading to a larger loss value in~\eqref{prob:kl-pc_convex}.

Suppose $y^\mode{n}_{j_n}\neq0$. Setting the derivative of the Lagrangian with respect to $p^\mode{n}_{j_n}$ equal to zero, we have
\[
-\frac{y^\mode{n}_{j_n}}{p^\mode{n}_{j_n}}
- q^\mode{n}_{j_n} + \xi^\mode{n} = 0.
\]
As we explained, $q^\mode{n}_{j_n}=0$ according to complementary slackness. Therefore,
\[
p^\mode{n}_{j_n}=y^\mode{n}_{j_n}/\xi^\mode{n}.
\]
The dual variable $\xi^\mode{n}$ should be chosen so that the equality constraint $\ones^\T\p^\mode{n}=1$ is satisfied. Together with our argument that $p^\mode{n}_{j_n}=0$ if $y^\mode{n}_{j_n}=0$, we come to the conclusion that
\begin{equation*}
p^\mode{n}_{j_n} = \frac{y^\mode{n}_{j_n}}{\ones^\T\y^\mode{n}}
= \frac{y^\mode{n}_{j_n}}{M}.
\end{equation*}

The result we derived in this section is summarized in the following theorem.
\begin{theorem}\label{thm:kl-pc}
The KL principal component of a nonnegative tensor $\Y$, i.e., the solution to Problem~\eqref{prob:kl-pc}, is
\begin{equation}\label{eq:kl-pc}
\lambda=M~~\text{and}~~
p^\mode{n}_{j_n} = \frac{y^\mode{n}_{j_n}}{M},
\end{equation}
where $M$ and $y^\mode{n}_{j_n}$ are defined in~\eqref{eq:M} and~\eqref{eq:y}, respectively.
\end{theorem}

Now let us take a deeper look at the solution we derived for Problem~\eqref{prob:kl-pc}. Suppose the data tensor $\Y$ is generated by drawing from the joint distribution $\prob[X_1,...,X_N]$ $M$ times. Our definition of $\y^\mode{n}$ in~\eqref{eq:y} essentially summarizes the number of times each possible outcomes of $X_n$ occurs, regardless of the outcomes of the other random variables. As a result, the optimal KL principal component factor $\p^\mode{n}$ is in fact the maximum likelihood estimate of the marginal distribution $\prob[X_n]$. On hindsight, the simple solution we provide for KL principal component in Theorem~\ref{thm:kl-pc} becomes very natural. The case when $K=1$ means the latent variable $\varXi$ can only take one possible outcome with probability one, which means $\varXi$ is not random. In other words, we are basically assuming that $X_1,...,X_N$ are \emph{independent} from each other. As a result, the joint distribution factors into the product of the marginal distributions
\[
\prob\left[X_1,...,X_N\right] = \prod_{n=1}^N \prob[X_n],
\]
and the marginal distributions can be simply estimated by ``marginalizing'' the observations, collected in $\Y$, and then normalizing to sum up to one.
This is elementary in probability. However, in the context of finding the principal component of a nonnegative \emph{tensor} using generalized KL-divergence, it is not at all obvious. Furthermore, the argument based on categorical random variables only applies to nonnegative \emph{integer} data, whereas our derivation of the KL principal component of a tensor is not restricted to integers or rational numbers, but works for general real nonnegative numbers as well.

\section{KL Approximation with Higher Ranks}\label{sec:em}

%It is encouraging to find one case for the KL-NCP problem, or even just a general tensor problem, that is not NP-hard and admits a very simple solution. On hindsight, we understand that the solution amounts to simply estimating the marginal distributions in the context of multi-variate multinomial probabilities. 

If $K>1$ in Problem~\eqref{prob:kl-ncp0}, there is more than one term in the logarithm; as a result, the nice transformation from~\eqref{prob:kl-pc} to~\eqref{prob:kl-pc_convex} cannot be directly applied. There is a way to apply something similar, and that is through the use of Jensen's inequality~\cite{jensen1906fonctions} (applied to the $-\log$ function)
\[
-\log\mathbb{E}\z \leq - \mathbb{E}\log\z.
\]
We use this inequality to define majorization functions for the design of an iterative upperbound minimization algorithm~\cite{razaviyayn2013unified}.

Suppose at the end of iteration $t$, the obtained updates are $\blambda^t$ and $\{\P^{\mode{n}\,t}\}$. At the next iteration, we define
\begin{equation}\label{eq:psi}
\varPsi_{j_1...j_Nk}^{t} = \lambda_k^t\prod_{n=1}^{N}P^{\mode{n}\,t}_{j_nk} \bigg/
\sum_{\kappa=1}^{K}\lambda_\kappa^t\prod_{n=1}^{N}P^{\mode{n}\,t}_{j_n\kappa}.
\end{equation}
According to this definition, it is easy to see that $\varPsi_{j_1...j_Nk}^{t}\geq0$ and $\sum_{k=1}^K\varPsi_{j_1...j_Nk}^{t}=1$. Assuming $\varPsi_{j_1...j_Nk}^{t}>0$, we have
\begin{align}\label{eq:upperbound}
&-Y_{j_1...j_N}\log\sum_{k=1}^{K}\lambda_k\prod_{n=1}^{N}P^\mode{n}_{j_nk} \nonumber\\
&= -Y_{j_1...j_N}\log\sum_{k=1}^{K}\frac{\varPsi_{j_1...j_Nk}^{t}}{\varPsi_{j_1...j_Nk}^{t}}
\lambda_k\prod_{n=1}^{N}P^\mode{n}_{j_nk} \nonumber\\
&\leq -Y_{j_1...j_N}
\sum_{k=1}^{K}\varPsi_{j_1...j_Nk}^{t}\log\lambda_k\prod_{n=1}^{N}P^\mode{n}_{j_nk} + \text{const.} 
\end{align}
Furthermore, equality is attained if $\blambda=\blambda^t$ and $\P^\mode{n} = \P^{\mode{n}\,t}, \forall n\in[N]$. This defines a majorization function for iteration $t+1$, and the minimizer of~\eqref{eq:upperbound} is set to be the update of this iteration. Since~\eqref{eq:upperbound} and the loss function of~\eqref{prob:kl-ncp0} are both smooth, the convergence result of the successive upperbound minimization (SUM) algorithm~\cite{razaviyayn2013unified} can be applied to establish that this procedure converges to a stationary point of Problem~\eqref{prob:kl-ncp0}. We should stress again that this procedure is made easy only after the multi-linear term in~\eqref{prob:kl-ncp} is equivalently replaced by the sum of the diagonal loadings in~\eqref{prob:kl-ncp0} through our careful problem formulation. Otherwise, the multi-linear term still remains, which together with~\eqref{eq:upperbound} does not end up being a simpler function to optimize.

The majorization function~\eqref{eq:upperbound} is nice, not only because it is convex, but also since it is reminiscent of the loss function~\eqref{prob:kl-pc_convex} when $K$ is equal to one---it decouples the variables down to the canonical components, i.e., $\lambda_k$ and the $k$-th columns of $\P^\mode{1},...,\P^\mode{N}$; each of the sub-problems takes the form of~\eqref{prob:kl-pc_convex}, replacing $Y_{j_1...j_N}$ with $Y_{j_1...j_N}\varPsi_{j_1...j_Nk}^{t}$. As a result, the update for iteration $t+1$ boils down to something similar to what we have derived in the previous iteration. Specifically, define
\[
M^t_k = \sum_{\substack{j_1,...,j_N}} Y_{j_1...j_N}\varPsi_{j_1...j_Nk}^{t}
\]
and
\[
y^{\mode{n}\,t}_{j_nk} = \sum_{\substack{j_1,...,j_{n-1},\\j_{n+1},...,j_N}} Y_{j_1...j_N}\varPsi_{j_1...j_Nk}^{t},
\]
then
\begin{equation}\label{eq:em}
\lambda_k^{t+1} = M_k^t
\quad\text{and}\quad
P_{j_nk}^{\mode{n}\,t+1} = y^{\mode{n}\,t}_{j_nk} \Big/ M_k^t.
\end{equation}

A nice probabilistic interpretation can be made to understand this algorithm. In the context of probabilistic latent variable modeling, the conditional distributions $\prob[X_n|\varXi]$ can be easily estimated if the joint observation $X_1=j_1,...,X_N=j_N$, \emph{and} $\varXi=k$ is given, because then we can collect all the observations with $\varXi=k$ and use the techniques derived in the previous section. 
Now that we do not observe $\varXi$, what we can do is to try to estimate $\prob[\varXi|X_1,...,X_N]$ instead. The $\varPsi$ defined in~\eqref{eq:psi} does exactly the job using the current estimate of $\blambda$ and $\{\P^\mode{n}\}$. Using this estimated posterior distribution of $\varXi$, we have a guess of the portion of $Y_{j_1...j_N}$ that jointly occurs with $\varXi=k$, and use that to obtain a new estimate of $\prob[\varXi]$ and $\prob[X_n|\varXi]$.
This is exactly the idea behind the expectation-maximization (EM) algorithm~\cite{dempster1977maximum}. Almost the same algorithm has been derived by Shashanka \textit{et al.}~\cite{Shashanka2008}, and the special case when $N=2$ is the EM algorithm for probabilistic latent semantic analysis (pLSA)~\cite{Hofmann1999}. However, we should mention that these algorithms were originally derived in the context of multinomial latent variable modeling, and without the help of Theorem~\ref{thm:equivalence}, it was not previously known that they can be used for generalized KL-divergence fitting as well.

Computationally, although the definition of $\varPsi$ in~\eqref{eq:psi} helps us obtain simple expressions~\eqref{eq:em} and intuitive interpretations, we do not necessarily need to explicitly form them when implementing the method for memory/computation efficiency considerations. To calculate $\blambda^{t+1}$ and $\{\P^{\mode{n}\,t+1}\}$ as in~\eqref{eq:em}, we first define $\widetilde{\Y}^t$ such that
\[
\widetilde{Y}^t_{j_1...j_N} = Y_{j_1...j_N}\bigg/
\sum_{\kappa=1}^{K}\lambda_\kappa^t\prod_{n=1}^{N}P^{\mode{n}\,t}_{j_n\kappa}.
\]
This operation requires passing through the data once, and if $\Y$ is sparse, $\widetilde{\Y}^t$ has exactly the same sparsity structure. Then,
\begin{equation}\label{eq:lambda}
\lambda_k^{t+1} = \lambda_k^t \widetilde{\Y}^t {\times_1}\p_k^{\mode{1}\,t} ... {\times_N}\p_k^{\mode{N}\,t},
\end{equation}
where $\times_n$ denotes the $n$-mode tensor-vector multiplication~\cite{Bader2007}, and $\p_k^{\mode{n}\,t}$ denotes the $k$-th column of $\P^{\mode{n}\,t}$. As for the factor matrices,
\begin{equation}\label{eq:factor}
\P^{\mode{n}\,t+1} = \P^{\mode{n}\,t} \ast \textsc{Mttkrp}\left(\widetilde{\Y}^t, \{\P^{\mode{\nu}\,t}\},n\right)
\end{equation}
followed by column normalization to satisfy the sum-to-one constraint,
where $\ast$ denotes matrix Hadamard (element-wise) product, and \textsc{Mttkrp} stands for the $n$-mode matricized tensor times Khatri-Rao product of all the factor matrices $\{\P^{\mode{\nu}\,t}\}$ except the $n$-th one.

It is interesting to notice that the update rules~\eqref{eq:lambda} and~\eqref{eq:factor} 
are somewhat similar to the widely used multiplicative-update (MU) rule for NMF~\cite{lee2001algorithms} and NCP~\cite{Chi2012}. 
The big difference lies in the fact that MU updates the factors \emph{alternatingly}, whereas EM updates the factors \emph{simultaneously}. This makes the EM algorithm extremely easy to parallelize---for $N$-way factorizations, we can simply take $N$ cores, each taking care of the computation for the $n$-th factor, and we can expect an almost $\times N$ acceleration, if the sizes of all the modes are similar.

\section{Illustrative Example}

\begin{figure}[t]
\centering
\hspace*{-20pt}
\includegraphics[width=1.1\linewidth]{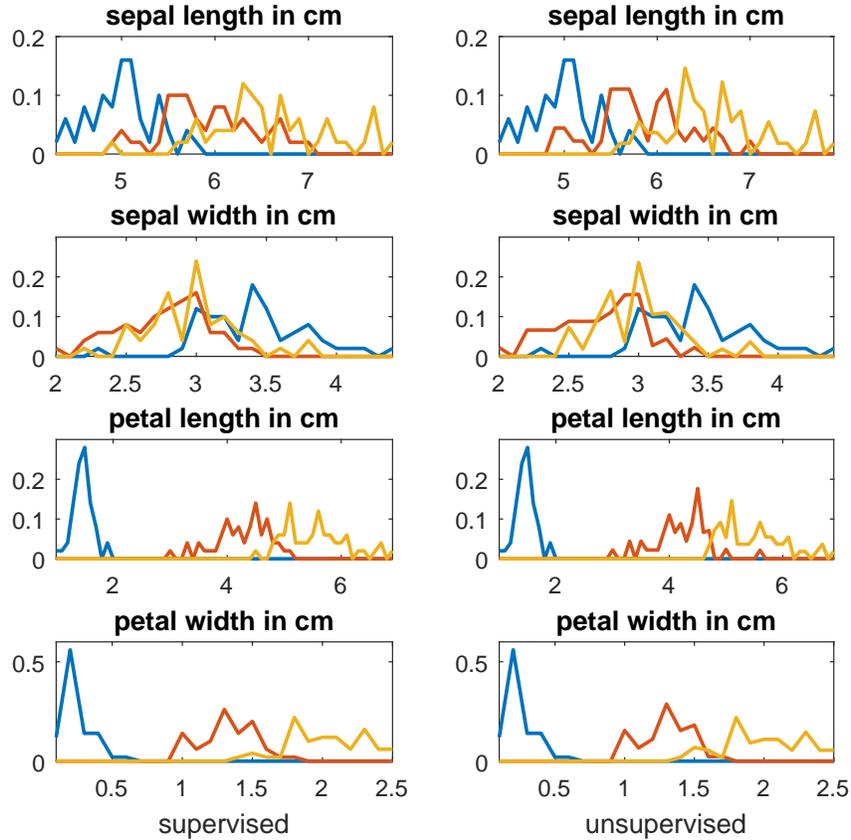}
\vspace{-20pt}
\caption{Learned $\mathbb{P}$[feature$|$class] in three colors for the three iris classes.
Left: class label is given. Right: class label not given.}
\label{fig:iris}
\end{figure}

We give an illustrative example using the Iris data set downloaded from the UCI Machine Learning repository. The data set contains 3 classes of 50 instances each, where each class refers to a type of iris plant. For each sample, four features are collected: sepal length/width and petal length/width in cm. The measurements are discretized into 0.1cm intervals, and the four features range between $4.3\sim7.9$, $2.0\sim4.4$, $1.0\sim6.9$, and $0.1\sim2.5$cm, respectively. We therefore form a 4-way tensor with dimension $37\times25\times60\times25$; for a new data sample, the corresponding entry in that tensor is added with one. %From Nikos: what does "is added with one" mean? 

Suppose the four features are conditionally independent given the class label. We can then collect all the samples from the same class into one tensor, and invoke Theorem~\ref{thm:kl-pc} to estimate their individual conditional distributions $\mathbb{P}$[feature$|$class], which are shown on the left panel of Fig.~\ref{fig:iris}. This learned conditional distribution can then be used to classify new samples, which is the idea behind the naive Bayes classifier~\cite{ng2002discriminative}.

Now suppose the class labels are \emph{not} given to us. We then collect \emph{all} the data samples into a single tensor. We still assume that the features are conditionally independent given the class label, even though the class label is now latent (unobserved). As per our discussion in Section~\ref{sec:em}, we can still try to estimate the conditional distribution using the EM algorithm~\eqref{eq:em}. We run the EM algorithm from multiple random initializations, and the result that gives the smallest loss is shown on the right of Fig.~\ref{fig:iris}. 

The astonishing observation is that the learned conditional distribution is almost identical to the one learned when the class label is given to us. This suggests that the conditional independence between features given class labels is actually a reasonable assumption in this case, contrary to the common belief that naive Bayes is an ``over-simplified'' model. It is also interesting to notice that this nice result is obtained using only 150 data samples, which is extremely small considering the size of the tensor. 

\section{Conclusion}
We studied the nonnegative CPD problem with generalized KL-divergence as the loss function. 
The most important result of this paper is the discovery that finding the KL principal component of nonnegative tensor is \emph{not} NP-hard. To make matters nicer, we derived a very simple closed-form solution for finding the KL principal component. This is a surprisingly pleasing result, considering that in the field of tensors ``most problems are NP-hard''~\cite{hillar2013most}.
Borrowing the idea for finding the KL principal component, an iterative algorithm for higher rank KL approximation was also derived, which is guaranteed to converge to a stationary point and is easily and naturally parallelizable.

\bibliographystyle{IEEEtran}
\bibliography{refs}

\end{document}